\documentclass[a4paper]{article}

\usepackage[utf8]{inputenc}
\usepackage{hyperref}
\usepackage{amsmath,amssymb,amsthm}
\usepackage{braket}
\usepackage{tikz}
\usepackage{longtable}
\usepackage{graphicx}
\usepackage{subcaption}

\newtheorem{theorem}{Theorem}[section]

\newtheorem{lemma}[theorem]{Lemma}
\newtheorem{definition}[theorem]{Definition}
\newtheorem{remark}[theorem]{Remark}
\newtheorem{example}[theorem]{Example}

\numberwithin{equation}{section}

\title{Perfect State Transfer in Weighted Cubelike Graphs}
\author{Jaideep Mulherkar \\ jaideep\_mulherkar@daiict.ac.in\and Rishikant Rajdeepak \\201521006@daiict.ac.in \and V. Sunita \\ v\_suni@daiict.ac.in}

\begin{document}
\maketitle
\begin{abstract}
    A continuous-time quantum random walk describes the motion of a quantum mechanical particle on an underlying graph. The graph itself is associated with a Hilbert space of dimension equal to the number of vertices. The dynamics of the walk is governed by the unitary operator $\mathcal{U}(t) = e^{iAt}$, where $A$ is the adjacency matrix of the graph. An important notion in the quantum random walk is the transfer of a quantum state from one vertex to another. If the fidelity of the transfer is unity, we call it a perfect state transfer. Many graph families have been shown to admit PST or periodicity, including cubelike graphs. These graphs are unweighted. In this paper, we generalize the PST or periodicity of cubelike graphs to that of weighted cubelike graphs. We characterize the weights for which they admit PST or show periodicity, both at time $t=\frac{\pi}{2}$.
\end{abstract}

\tableofcontents

\section{Introduction}
Graphs are among the most powerful mathematical tools used to model many practical problems in basic sciences, computer sciences, social sciences, network sciences etc. One of the most efficient method to study problems in graph theory is by using linear algebra and group theory. A graph can be represented by matrices, such as adjacency matrix and Laplacian matrix, whose spectra produce several properties of the graph. These algebraic methods answer questions related to (a) graph isomorphism, (b) network flow problems, (c) graph coloring, (d) transitivity (d) regularity, (e) tree number, and so on~\cite{biggs1974,godsilbook}. One of the recent applications of graph spectra is related to quantum physics, where a graph is used to model a network of quantum particles and the evolution of these particles are described by the adjacency matrix of the graph. This phenomenon is known as continuous-time quantum random walk. 

A quantum random walk is a generalization of a classical random walk. It is a very important tool for designing quantum algorithms because it is universal for quantum computation~\cite{feynman1986,childs2009,childs2013}. The quantum walk is of two types: discrete and continuous. In discrete case, a graph is associated with a Hilbert space of dimension $N\times \Delta$, where $N$ is the number of vertices and $\Delta$ is the maximum degree of the graph. The discrete-time quantum walk is described by an unitary operator that repeatedly acts on the associated Hilbert space. In continuous case, the graph is associated with the Hilbert space of dimension $N$ and the evolution of the system is described by $e^{\iota tA}$, where $A$ is the adjacency matrix of the graph and $t$ is real time. Some of the works done on quantum walks can be found in~\cite{aharonov1993,farhi1998,ambainis2001,aharonov2001,moore2002,childs2002,shenvi2003,childs2003,kempe2003,julia2003}.

An important feature of a quantum walk is the transfer of quantum state from one vertex to another with high fidelity. It has been found that many graph families, mostly in continuous case, allow transmission of quantum states with fidelity equal to unity, i.e., the transfer is perfect. For details, see~\cite{bose2003,christandl2005,bernasconi2008,cheung2011,godsil2011,godsil2012,godsil2012(2),stefanak2016,cao2019}. Among these graphs, cubelike graphs are most famous one whose properties have been characterized for determining the existence and finding the pair of vertices admitting perfect state transfer in constant time. It is to be noted that all cubelike graphs do not allow perfect state transfer. But, all cubelike graphs are periodic with period dividing $\frac{\pi}{2}$. In fact, all integral graphs are periodic and cubelike graphs are integral graphs.

In this report we make the following contributions:
\begin{enumerate}
\item We present a characterization of graphs whose adjacency matrices have same eigenvectors. As a special case we characterize cubelike graphs. Eigenvectors for cubelike graphs have been derived using representation theory of finite groups~\cite{babai1979,benjamin2012}. We present an alternate proof using only basic linear algebra. Our motivation was to study some properties which are shared by graphs having same eigenvectors. For example, all cubelike graphs of dimension $n$, $n\in\mathbb{Z}^+$, have same eigenvectors and they have common properties such as; they are all regular, have Hamiltonian path, are vertex and edge transitive, and so on. 
\item We use the above characterization to study perfect state transfer on graphs having eigenvectors same as unweighted cubelike graphs. We show that weighted cubelike graphs with integer weights, and having eigenvectors same as unweighted cubelike graphs, admit perfect state transfer in time $\frac{\pi}{2}$, or they are periodic with period $\frac{\pi}{2}$. 
\end{enumerate}

\section{Preliminaries}
A simple graph $\Gamma$ is represented by a pair $(V,E)$, where $V$ is a non-empty set whose elements are called vertices or nodes, and $E$ is a collection of unordered pairs of distinct vertices. There are several variations of graphs such as multigraph, pseudograph, directed graph, and so on. We focus on undirected weighted graph which is described by the triplet $(V,E,f)$, where the pair $(V,E)$ represents a graph and the function $f:E\rightarrow\mathbb{R}$ assigns weight to each edge. If the range of the function $f$ is $\{0,1\}$, then the graph is a simple graph. 

A graph $\Gamma=(V,E,f)$ on $n$ vertices is described by its adjacency matrix $A$. Suppose vertices in $\Gamma$ are represented by numbers $\{1,\dots, n\}$, then $(i,j)$-entry of $A$ is the weight of the edge $(i,j)$, i.e., $A_{i,j}=f((i,j))$. If $\Gamma$ is simple then
\[
A_{i,j} = \begin{cases} 1,&\mbox{if }(i,j)\in E \\ 0, & \mbox{if }(i,j)\notin E \end{cases}.
\]
The adjacency matrix $A$ is real-symmetric matrix which can be decomposed into its eigenvectors with real eigenvalues. We now define cubelike graphs, see Fig.~\ref{fig:cubelike}

\begin{figure}[t]
\tiny
\centering
    \begin{subfigure}[t]{.45\textwidth}
    \begin{tikzpicture}[scale=.7]
        \begin{scope}[every node/.style = {draw, circle, inner sep=1pt}]
        \node (0) at (-1,1) [label=left:$000$]{$1$};
        \node (1) at (1,1) [label=right:$001$]{$2$};
        \node (3) at (1,-1) [label=right:$011$]{$4$};
        \node (2) at (-1,-1) [label=left:$010$]{$3$};
        \node (4) at (-3,3) [label=left:$100$]{$5$};
        \node (5) at (3,3) [label=right:$101$]{$6$};
        \node (7) at (3,-3) [label=right:$111$]{$8$};
        \node (6) at (-3,-3) [label=left:$110$]{$7$};
    \end{scope}
        \draw (0) edge  (1); 
        \draw (0) edge  (2); 
        \draw (0) edge  (4); 
        \draw (1) edge  (3); 
        \draw (1) edge  (5); 
        \draw (2) edge  (3);
        \draw (2) edge  (6); 
        \draw (3) edge  (7); 
        \draw (4) edge  (5); 
        \draw (4) edge  (6); 
        \draw (5) edge  (7); 
        \draw (6) edge  (7); 
    \end{tikzpicture}
    \caption{\label{fig:Q3} $\mathcal{Q}_3$}
    \end{subfigure}
    \begin{subfigure}[t]{.45\textwidth}
    \begin{tikzpicture}[scale=.7]
        \begin{scope}[every node/.style = {draw, circle, inner sep=1pt}]
        \node (0) at (-1,1) [label=above:]{$1$};
        \node (1) at (1,1) [label=above:]{$2$};
        \node (3) at (1,-1) [label=below:]{$4$};
        \node (2) at (-1,-1) [label=below:]{$3$};
        \node (4) at (-3,3) [label=above:]{$5$};
        \node (5) at (3,3) [label=above:]{$6$};
        \node (7) at (3,-3) [label=below:]{$8$};
        \node (6) at (-3,-3) [label=below:]{$7$};
    \end{scope}
        \draw (0) edge (1); 
        \draw (0) edge (2); 
        \draw (0) edge  (4);
        \draw (0) edge[bend left=20] (3);
        \draw (0) edge[bend right=20] (7);
        \draw (1) edge  (3); 
        \draw (1) edge (5);
        \draw (1) edge[bend right=20]  (2);
        \draw (1) edge [bend left=20]  (6);
        \draw (2) edge  (3);
        \draw (2) edge  (6);
        \draw (2) edge[bend right=20] (5);
        \draw (3) edge[bend left=20] (4);
        \draw (3) edge   (7); 
        \draw (4) edge   (5);
        \draw (4) edge  (6);
        \draw (4) edge[bend right=45] (7);
        \draw (5) edge  (7); 
        \draw (5) edge[ bend left=45]  (6);
        \draw (6) edge  (7); 
    \end{tikzpicture}
    \caption{\label{fig:AQ3} $\mathcal{AQ}_3$}
    \end{subfigure}
    \caption{Pictorial representation of (\subref{fig:Q3}) a Hypercube of dimension $3$ $\mathcal{Q}_3=Cay(\mathbb{Z}_2^3,\{001,010,100\})$ and (\subref{fig:AQ3}) an Augmented cube of dimension $3$ $\mathcal{AQ}_3=Cay(\mathbb{Z}_2^3,\{001,010,011,100,111\})$.}
    \label{fig:cubelike}
\end{figure}
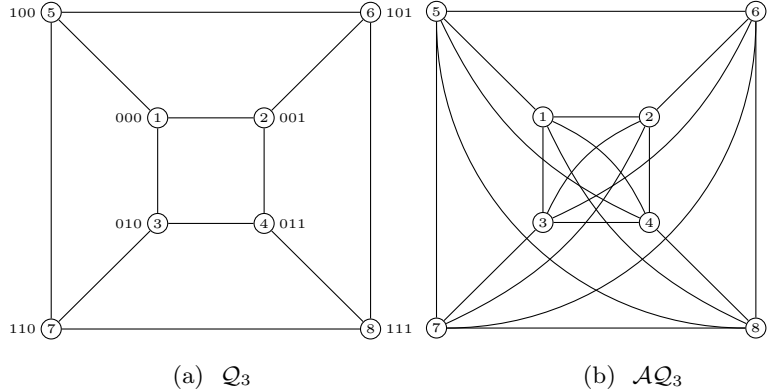

\begin{definition}[Cubelike Graphs]
A Cayley graph is a graph defined over a pair $(G,\Omega)$, where $G$ is a finite group and $\Omega$ is a generating set of $G$, with the properties; (1) $\Omega$ does not contain the identity element, and (2) $\Omega$ is closed under the group inverse, i.e., $x^{-1}\in\Omega$ for all $x\in \Omega$. The Cayley graph, denoted by $Cay(G,\Omega)$, is a graph whose vertices are identified with the elements of $G$ and the edge set is $\{(x,y):xy^{-1}\in\Omega\}$. Clearly, $Cay(G,\Omega)$ is a regular graph on $N=\mid G\mid$ vertices with regularity $\Delta=\mid\Omega\mid$. The Cayley graph $Cay(\mathbb{Z}_2^n,\Omega)$ defined over the Boolean group $\mathbb{Z}_2^n$ is called cubelike graph of dimension $n$.
\end{definition}

\subsection{Spectral decomposition}
An $n\times n$ normal matrix $N\in\mathbb{C}^n\times \mathbb{C}^n$ is defined by 
\[
NN^\dagger = N^\dagger N,
\]
where $N^\dagger$ is the adjoint (complex-conjugate transpose) of $N$. If $NN^\dagger=I$, then $N$ is called unitary matrix. Suppose $N=N^\dagger$ then $N$ is Hermitian (also known as self-adjoint). If entries in $N$ are real and $N$ is Hermitian then $N$ is a real-symmetric matrix, i.e., $N=N^T$. The spectral theory diagonalizes a normal matrix as stated below.
\begin{theorem}[Spectral theory]\label{thm:spectral}\cite{hoffman2004}
Let $N$ be an $n\times n$ normal matrix. Then, the following statements hold true.
\begin{enumerate}
    \item There is an $n\times n$ unitary matrix $P$ such that $D=P^{-1}NP$ is diagonal. 
    \item Suppose $\lambda_1,\dots,\lambda_m$ are distinct eigenvalues of $N$, and $E_j$ is the orthogonal projection on the eigenspace associated with $\lambda_j$, then the spectral decomposition (also called spectral resolution) of $N$ is given by
    \begin{equation}
        N = \lambda_1E_1+\cdots +\lambda_mE_m.
    \end{equation}
\end{enumerate}
\end{theorem}
The set $S=\{\lambda_1,\dots,\lambda_m\}$ of eigenvalues is called the spectrum of $N$. In Theorem~\ref{thm:spectral}, the $j$-th column of $P$, denoted by $P_{*j}$, is an eigenvector of $N$ with eigenvalue $d_j=D_{j,j}$, i.e., $NP_{*j}=d_jP_{*j}$. The spectral theory for unitary, Hermitian and real-symmetric matrices can be formulated as;
\begin{theorem}\cite{hoffman2004}
If $N$ is a normal matrix with unitary matrix $P$ satisfying $D=P^{-1}DP$, such that $D$ is diagonal, then,
\begin{enumerate}
    \item $N$ is self-adjoint iff its eigenvalues are real, 
    \item $N$ is unitary iff its eigenvalues are of absolute value $1$.
    \item $P$ is orthogonal and $D$ is real iff $N$ is real-symmetric matrix.
\end{enumerate}
\end{theorem}

Another useful result from linear algebra associates spectra of normal matrix $N$ with its matrix exponential $e^{N}$.
\begin{theorem}\cite{hoffman2004}
Let $N$ be a normal matrix with unitary matrix $P$ satisfying $D=P^{-1}NP$, such that $D$ is diagonal. Let $N=\sum_{j=1}^m\lambda_jE_j$ be the spectral decomposition of $N$. Suppose $f$ is a complex-valued function defined over $S$, then the following statements hold true.
\begin{enumerate}
    \item The linear operator $f(N)$ defined by
    \[
    f(N) = \sum_{j=1}^mf(\lambda_j)E_j
    \]
is a diagonalizable normal operator with spectrum $f(S)$. In other words,
\[f(D)=P^{-1}f(N)P\]
is diagonal with the same unitary matrix $P$, i.e., the $j$-th column $P_{*j}$ of $P$ is eigenvector of $f(N)$ with eigenvalue $f(d_j)$, $1\leq j\leq n$. 
\end{enumerate}
\end{theorem}

\subsection{Perfect state transfer}
Let $\Gamma$ be an undirected and weighted graph with loops and $A$ be the adjacency matrix. A quantum walk on $\Gamma$ is described by an evolution of the quantum system associated with the graph. Suppose the graph has $N$ vertices, then it is associated with a Hilbert space $\mathcal{H}_P\cong \mathbb{C}^N$, called the position space, and the computational basis for $\mathcal{H}_P$ is represented by;
\[
\{\ket{v}:v\mbox{ is a vertex in }\Gamma\}.
\]
The continuous-time quantum walk on $\Gamma$ is described by the transition matrix $\mathcal{U}(t)=e^{\iota t A}$, where $\iota=\sqrt{-1}$, i.e., if $\ket{\psi(0)}$ is the initial state of the system then, at time $t$, the state of the system is given by 
\[
\ket{\psi(t)}=e^{\iota t A}\ket{\psi(0)}.
\]
\begin{definition}
A graph is said to admit perfect state transfer, if the quantum walker beginning at some vertex $u$ reaches a distinct vertex $v$ with probability $1$, i.e., $\mathcal{U}(\tau)\ket{u}=\lambda\ket{v}$, for some $\lambda\in\mathbb{C}$, satisfies
\[
\braket{v|e^{\iota \tau A}|u}=|\lambda|^2=1,\qquad \mbox{ for some real time }\tau.
\]
Alternatively, we say perfect state transfer occurs from the vertex $u$ to the vertex $v$.
\end{definition}
\begin{remark}
The matrix exponential of $\iota t A$ can be seen as
\[
\mathcal{U}(t)=\sum_{k=0}^{\infty}\iota^k\frac{t^kA^k}{k!}.
\]
Since $A$ is symmetric $\mathcal{U}(t)$ is symmetric, and $\overline{e^{\iota t A}}=e^{-\iota t A}$ implies $\mathcal{U}(t)$ is unitary. Moreover,
\[
\mathcal{U}(t_1+t_2)=e^{\iota(t_1+t_2)A}=e^{\iota t_1A}e^{\iota t_2A}=\mathcal{U}(t_1)\mathcal{U}(t_2).
\]
\end{remark}

As a special case of the spectral theory (Theorem~\ref{thm:spectral}), we express the transition matrix $\mathcal{U}(t)$ in more useful way. 
\begin{lemma}[Spectral theory]
Let $\Gamma$ be a graph on $n$ vertices with the adjacency matrix $A$. Suppose $P$ is an orthogonal matrix such that $D=P^TAP$ is real diagonal, then the transition matrix $\mathcal{U}(t)=e^{\iota t A}$ is expressed as
\begin{equation}
    \begin{split}
        \mathcal{U}(t) &= Pe^{\iota t D}P^T \\
        &= \sum_{k=1}^ne^{\iota t D_{k,k}} \ket{P_{*k}}\bra{P_{*k}} \\
        \implies \mathcal{U}(t)_{u,v} &= \sum_{k=1}^n e^{\iota t D_{k,k}}P_{u,k}P_{v,k}.
    \end{split}
\end{equation}
\end{lemma}

\begin{example}
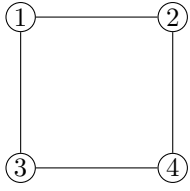
\begin{figure}[t]
    \centering
    \begin{tikzpicture}
        \begin{scope}[every node/.style={draw,circle,inner sep=1pt}]
            \node (0) at (-1,1) {1};
            \node (1) at (1,1) {2};
            \node (3) at (1,-1) {4};
            \node (2) at (-1,-1) {3};
        \end{scope} 
        \draw (0) -- (1); \draw (0) -- (2); \draw (1) -- (3); \draw (2) -- (3);
    \end{tikzpicture}
    \caption{PST occurs between the pairs \{1,4\} and \{2,3\} with time $\frac{\pi}{2}$, and the graph is periodic with period $\pi$.}\label{fig:cycle}
    \label{fig:my_label}
\end{figure}
Consider the graph on cycle of size 4 (Figure~\ref{fig:cycle}). Then, the adjacency matrix $A$ is given by
\[
A=\begin{bmatrix} 0 & 1 & 1 & 0 \\ 1 & 0 & 0 & 1 \\ 1&0&0&1 \\ 0&1&1&0  \end{bmatrix},
\]
with spectral decomposition 
\[
A=PDP^T,\mbox{ where }P=\frac{1}{2}\begin{bmatrix} 1&1&1&1 \\ 1&-1&1&-1 \\ 1&1&-1&-1 \\ 1&-1&-1&1 \end{bmatrix}, \mbox{ and }D=\begin{bmatrix} 2&0&0&0 \\ 0&0&0&0 \\ 0&0&0&0 \\ 0&0&0&-2 \end{bmatrix}.
\]
Therefore, the spectral decomposition for the transition matrix with time $t=\pi/2$ is
\[
\mathcal{U}(t=\pi/2) = \sum_{k=1}^4e^{\iota \frac{\pi}{2} D_{k,k}} \ket{P_{*k}}\bra{P_{*k}} = \begin{bmatrix} 0&0&0&-1 \\ 0&0&-1&0 \\0&-1&0&0 \\ -1&0&0&0 \end{bmatrix}.
\]
Thus, perfect state transfer occur between the pairs $(1,4)$ and $(2,3)$, both in time $\frac{\pi}{2}$.
\end{example}

\begin{lemma}
If perfect state transfer occurs from a vertex $u$ to a vertex $v$ in a graph, then perfect state transfer occurs from $v$ to $u$.
\end{lemma}
\begin{proof}
Suppose for some real $\tau$, $\mathcal{U}(\tau)\ket{u}=\lambda \ket{v}$, for some scalar $\lambda$, with $|\lambda|=1$. Since $\mathcal{U}(\tau)$ is symmetric, 
\[
\mathcal{U}(\tau)\ket{v}=\lambda\ket{u}.
\]
\end{proof}

Notice that $\mathcal{U}(\tau)^2\ket{u}=\lambda^2\ket{u}$. But, $\mathcal{U}(\tau)^2=\mathcal{U}(2\tau)$. Therefore, the walker returns back to the starting vertex $u$ after time $2\tau$. Such property is termed as periodicity.
\begin{definition}
If a walker begins at a vertex $u$ in a graph and returns back after time $\tau$ with one probability, i.e.,
\[
\mathcal{U}(\tau)\ket{u}=\lambda\ket{u},\qquad |\lambda|=1,
\]
we say the graph is periodic at $u$ or it is periodic relative to $u$, with period $\tau$. If the graph is periodic relative to every vertex with the same period $\tau$, we say the graph is periodic.
\end{definition}
\begin{lemma}
If a graph admits perfect state transfer from vertex $u$ to $v$, then the graph is periodic at $u$.
\end{lemma}
\begin{lemma}
If perfect state transfer occurs between $u$ and $v$, and between $u$ and $w$, then $v=w$.
\end{lemma}
There are graphs which are periodic at a vertex, or periodic at each vertex but do not admit perfect state transfer. For example, star $K_{1,n}$ is periodic at one vertex and complete graph $K_n$ is periodic, but none of them has PST pairs. In general, simple graphs do not admit perfect state transfer. However, by assigning weights to edges some of them may show the occurrence of perfect state transfer. In \cite{canul2009}, it is shown that the join of a weighted two-vertex graph with any regular graph has perfect state transfer. For more reference on weighted graphs allowing PST, see \cite{canul2010,ge2011,vinet2020}. 

\section{Graphs having same eigenvectors}
Let $\Gamma$ be an undirected weighted graph wth adjacency matrix $A$. We say an orthogonal matrix $P$ and a column matrix $X$ are eigenvectors and eigenvalues of $\Gamma$, respectively, if columns of $P$ are eigenvectors of $A$ and entries in $X$ are corresponding eigenvalues, i.e., $AP_{*j}=X_{j,1}P_{*j}$, $1\leq j\leq n$, where $n$ is the number of vertices in $\Gamma$. We design a method, by using the property of spectral theory, to find other graphs having same eigenvectors along with their eigenvalues. As an application, we characterize weighted cubelike graphs having eigenvectors same as cubelike graphs.

\subsection{The characterization}
\begin{theorem}\label{thm:main}
Let $P$ be an $n\times n$ orthogonal matrix. Then, there exists a set of $n$ paired indices
\[
\mathcal{I}=\{(r_j,c_j):\;1\leq j\leq n,\;1\leq r_j,c_j\leq n\}
\]
such that 
\begin{enumerate}
    \item the $n\times n$ matrix $Q$, whose $j$-th row is the Hadamard product of $r_j$-th row and $c_j$-th row of $P$, denoted by \[Q_{j*}=P_{r_j*}\odot P_{c_j*},\qquad 1\leq j\leq n,\] is invertible,
    \item for each $Z\in\mathbb{R}^{n\times 1}$, there is a weighted graph having eigenvectors $P$ and eigenvalues $Q^{-1}Z$.
\end{enumerate}
In other words, the triplet $(P,\mathcal{I},Z)$, with fixed $P$, determines all graphs having eigenvectors $P$ along with its eigenvalues.
\end{theorem}
\begin{proof}
Define the matrix $S$ by
\[
S = \begin{bmatrix} P_{*1}\otimes P_{*1} & \cdots & P_{*n}\otimes P_{*n} \end{bmatrix},
\]
where the $j$-th column $S_{*j}=P_{*j}\otimes P_{*j}$ ($1\leq j\leq n$) is the tensor product of $j$-th column of $P$ with iteself. Since $P$ is orthogonal, the columns of $S$ form a set of $n$ linearly independent vectors. Thus, there is a set of $n$ linearly independent rows in $S$. Suppose rows corresponding to $n$ paired indices $\mathcal{I}=\{(r_j,c_j):\;1\leq j\leq n,\;1\leq r_j,c_j\leq n\}$ are linearly independent, then the submatrix $Q$ given by
\begin{equation}\label{eq:Q}
Q=\begin{bmatrix}
    P_{r_1,1}P_{c_1,1} & \cdots & P_{r_1,n}P_{c_1,n} \\
    \vdots & \ddots & \vdots \\
    P_{r_n,1}P_{c_n,1} & \cdots & P_{r_n,n}P_{c_n,n}
    \end{bmatrix},
\end{equation}
is invertible. Let $Z\in\mathbb{R}^{n\times 1}$ and $X=Q^{-1}Z$. We now construct a real-symmetric matrix $A$ associated with $Z$. Assume, entries in $Z$ correspond to $n$ entries in $A$ corresponding to the paired index set $\mathcal{I}$, viz.,
\[
Z = \begin{bmatrix}
A_{r_1,c_1} \\ \vdots \\ A_{r_n,c_n}
\end{bmatrix}
\] 
Then, the system of linear equations
\begin{equation}\label{eq:QX=Z}
    QX = Z
\end{equation}
describes the following system of linear equations
\begin{equation}\label{eq:SX=Y}
    SX = Y,\mbox{ where }Y=\begin{bmatrix} A_{*1} \\ \vdots \\ A_{*n} \end{bmatrix},\; A_{*j}=\begin{bmatrix} A_{1,j} \\ \vdots \\ A_{n,j} \end{bmatrix}.
\end{equation}    
That is, the solution for $X$ in Eq.~\ref{eq:QX=Z} and Eq.~\ref{eq:SX=Y} are same. Eq.~\ref{eq:SX=Y} is an alternate representation for $A=PDP^T$, where $D$ is diagonal with diagonal entries $X$, viz.,
\begin{equation}
\begin{split}
    A_{i,j} &= \sum_{k=1}^nX_{k,1}P_{i,k}P_{j,k},\quad 1\leq i,j\leq n \\
    &= \sum_{k=1}^nX_{k,1}P_{j,k}P_{i,k} \\
    &= A_{j,i}. \\
    \implies A &= PDP^T, \qquad\mbox{where } D_{k,k}=X_{k,1}.
\end{split}
\end{equation}
Thus, $A$ is determined uniquely for each $Z$. It is real and symmetric matrix with eigenvectors $P$ and eigenvalues $X=Q^{-1}Z$. In other words, the graph whose adjacency matrix is $A$ has eigenvectors $P$ and eigenvalues $Q^{-1}Z$.
\end{proof}

\begin{lemma}\label{lem:nonzerorow}
With the assumptions from Theorem~\ref{thm:main}, if a row of $P$ has no zero entry, say the $k$-th row, then the invertible matrix $Q$ is given by
\[
Q=\begin{bmatrix} P_{k,1}P_{*1} & \cdots & P_{k,n}P_{*n} \end{bmatrix}.
\]
\end{lemma}
\begin{proof}
Since columns of $P$ form linearly independent set, so if
\[
\sum_{j=1}^nc_j P_{k,j}P_{*j} = 0,
\]
then $c_jP_{k,j}=0$, for all $1\leq j\leq n$. But, $P_{k,j}\neq 0$ implies $c_j=0$. Hence, the columns of $Q$ are linearly independent, i.e., $Q$ is invertible.
\end{proof}
\begin{lemma}\label{lem:onerow}
With the assumptions from Theorem~\ref{thm:main}, if a row of $P$ has constant entries $\begin{bmatrix} \mu & \cdots & \mu \end{bmatrix}$, then the matrix $Q$, given by $Q=\mu P$, is an invertible submatrix of $S$, and the graph associated with $Z\in\mathbb{R}^{n\times 1}$ has eigenvalues $X=Q^{-1}Z=\frac{1}{\mu} P^TZ$.
\end{lemma} 
\begin{proof}
Since $P$ is invertible, $\mu\neq 0$. By Lemma~\ref{lem:nonzerorow}, $Q=\mu P$ is invertible and the result follows.
\end{proof}

\subsection{Characterization of weighted cubelike graphs}
For an application of Theorem~\ref{thm:main} to graph theory, notice that given an $n\times n$ orthogonal matrix $P$ we can construct a graph $\Gamma$ by giving partial information about it. In particular, for the $n$ paired indices $\mathcal{I}=\{(r_j,c_j):\;1\leq j\leq n,\;1\leq r_j,c_j\leq n\}$, if we assign weights to edges $(r_j,c_j)$, then all other edges can be generated by applying Eqs.~\ref{eq:QX=Z} and \ref{eq:SX=Y}. Here, we have assumed that vertices are labeled by integers $\{1,2,\dots,n\}$. The graphs obtained by this method are usually weighted graphs with loops. It is of interest to study the case where $P$ and $Z$ lead to unweighted graphs. For an instance, Lemma~\ref{lem:onerow} can be used to characterize cubelike graphs. It is to be noted that eigenvectors and eigenvalues have already been characterized for cubelike graphs using representation theory of finite groups~\cite{babai1979,benjamin2012}. We present the same result without using the representation theory.

\subsubsection{Cubelike graphs}
\begin{theorem}\label{thm:cubelike}
Let $n=2^d$, for some positive integer $d$. Let $P$ be an $n\times n$ orthogonal matrix, defined by
\begin{equation}\label{eq:P}
P_{i,j} = \frac{1}{\sqrt{n}}(-1)^{\braket{b(i)|b(j)}},\qquad 1\leq i,j\leq n,
\end{equation}
where $b(k)$ denotes the $k$-th binary string in $\mathbb{Z}_2^d$. Let $Z\in\mathbb{Z}_2^{n\times 1}$, such that 
\[\Omega=\{b(l):1\leq l\leq n,\;Z_{l,1}=1\}\]
is a linearly independent subset of $\mathbb{Z}_2^d$. Then, the graph associated with $Z$ is isomorphic to the cubelike graph $Cay(\mathbb{Z}_2^d,\Omega)$ with eigenvalues $\sqrt{n}PZ$.
\end{theorem}
\begin{proof}
Suppose $A$ is the adjacency matrix of the graph $\Gamma$ associated with $Z$, and put $X=\sqrt{n}PZ$ since $P=P^T$. Then,
\[
\begin{split}
A_{i,j} 
&= \sum_{k=1}^nX_{k,1}P_{i,k}P_{j,k} \\
&= \sum_{k=1}^n\left(\sqrt{n}\sum_{l=1}^nP_{l,k}Z_{l,1}\right)P_{i,k}P_{j,k} \\
&= \sqrt{n}\sum_{k=1}^n\sum_{l=1}^nZ_{l,1} \frac{(-1)^{\braket{b(l)|b(k)}}}{\sqrt{n}} \frac{(-1)^{\braket{b(i)|b(k)}}}{\sqrt{n}}\frac{(-1)^{\braket{b(j)|b(k)}}}{\sqrt{n}} \\
&= \frac{1}{n}\sum_{l=1}^n Z_{l,1} \left(\sum_{k=1}^n (-1)^{\braket{b(l)\oplus b(i) \oplus b(j)|b(k)}} \right) \\
&= \begin{cases} \frac{1}{n} \sum_{l=1}^n Z_{l,1} \;;&\mbox{if }b(l)\oplus b(i)\oplus b(j) = b(1), \\ \frac{1}{n}\sum_{l=1}^n Z_{l,1}\times 0 \;;& \mbox{otherwise}. \end{cases}
\end{split}
\]
Since $b(i)\oplus b(j)\in\mathbb{Z}_2^n$, we get $b(i)\oplus b(j)=b(h)$ for unique $l=h$. Therefore,
\[
A_{i,j} = \begin{cases} 1; & \mbox{if }l=h \mbox{ and }Z_{h,1}=1\\ 0; & \mbox{otherwise}. \end{cases}
\]
This is equivalent to say that $i$-th vertex is adjacent to $j$-th vertex only if $b(i)\oplus b(j)=b(l)$ and $Z_{l,1}=1$. Now, define a subset of $\mathbb{Z}_2^d$ by, $\Omega = \{b(l):\;1\leq l\leq n,\;Z_{l,1}=1\}$. Clearly, $\Gamma$ is isomorphic to the Cayley graph $Cay(\mathbb{Z}_2^d,\Omega)$ because the following statements are equivalent;
\begin{enumerate}
    \item Two vertices $i$ and $j$ are adjacent in $\Gamma$ only if $Z_{h,1}=1$, where $b(h)=b(i)\oplus b(j)$.
    \item Two vertices $b(i)$ and $b(j)$ are adjacent in $Cay(\mathbb{Z}_2^d,\Omega)$ only if $b(i)\oplus b(j)\in\Omega$.
\end{enumerate}
\end{proof}

\subsubsection{Weighted cubelike graphs}
The assignment of weights to cubelike graphs is such that its eigenvectors remain unchanged, which is represented by Eq.~\ref{eq:P} in Theorem~\ref{thm:cubelike}. This can be achieved by applying Theorem~\ref{thm:main} over eigenvectors of cubelike graphs. Let $Z\in\mathbb{R}^{n\times 1}$, then the associated graph has  eigenvalues $X=\sqrt{n}PZ$, by Lemma~\ref{lem:onerow}. We show that if entries in $Z$ are integers with its first entry zero, then the associated graph $\Gamma$ is isomorphic to a weighted cubelike graph, i.e., if $\Omega=\{b(l):1\leq l\leq n,\;Z_{l,1}\neq 0\}$ then two vertices $u$ and $v$ are adjacent in $\Gamma$ iff $b(u)$ and $b(v)$ are adjaceny in the Cayley graph $Cay(\mathbb{Z}_2^d,\Omega)$.

\begin{lemma}
Suppose $P$ is the matrix from Eq~\ref{eq:P} and $Z\in\mathbb{R}^{n\times 1}$. Let $\Gamma$ be the graph associated with $Z$.
\begin{enumerate}
    \item If $Z_{1,1}=0$, then $\Gamma$ has no loop.
    \item If $Z\in\mathbb{Z}^{n\times 1}$, then eigevalues of $\Gamma$ are integers with same parity.
\end{enumerate}
\end{lemma}
\begin{proof}
Suppose $A$ is the adjacency matrix of $\Gamma$. Then, 
\[
\begin{split}
    A_{i,i} &= \sum_{k=1}^nX_{k,1}P_{i,k}P_{i,k} = \frac{1}{n}\sum_{k=1}^nX_{k,1} \\
    &= \frac{1}{n}\sum_{k=1}^n\sum_{l=1}^nZ_{l,1}(-1)^{\braket{b(l)|b(k)}} \\
    &= \frac{1}{n}\sum_{l=1}^nZ_{l,1}\sum_{k=1}^n(-1)^{\braket{b(l)|b(k)}} \\
    &= \frac{1}{n}Z_{1,1} = 0 .
\end{split}
\]
Next, assume that $Z$ has integer entries. Then, 
\begin{equation}\label{eq:eigen}
X_{k,1} = \sum_{l=1}^n(-1)^{\braket{b(k)|b(l)}}Z_{l,1},\qquad 1\leq k\leq n,
\end{equation}
are integers. Since $\braket{b(1)|b(l)}=0$, $1\leq l\leq n$, $X_{1,1}$ is odd (or even) if the number of odd entries in $Z$ are odd (or even). The difference $\lambda_k-\lambda_1$ is even because it contains even number of entries from $Z$. Therefore, $\lambda_k$ and $\lambda_1$ have same parity.
\end{proof}

\section{Perfect state transfer}
Cubelike graphs of same dimension have a common eigenvectors and admit perfect state transfer or periodicity with time $\frac{\pi}{2}$. We give a general result for perfect state transfer in weighted cubelike graphs having same eigenvectors. Then after, we verify the result numerically.

\subsection{PST in weighted cubelike graphs}
\begin{theorem}\label{thm:pst}
Let $n=2^d$, for some positive integer $d$. Let $P$ be the $n\times n$ orthogonal matrix defined by
\[
P_{i,j}=\frac{1}{\sqrt{n}}(-1)^{\braket{b(i)|b(j)}}, \qquad 1\leq i,j\leq n.
\]
If $Z\in\mathbb{Z}^{n\times 1}$ with its first entry zero, then the associated graph admits perfect state transfer or is periodic, with time $\frac{\pi}{2}$ in either cases.
\end{theorem}
\begin{proof}
Suppose, there exists perfect state transfer between the pair $\{u,v\}$ with time $\tau$. Put $b(\sigma)=b(u)\oplus b(v)$, for some $1\leq \sigma \leq n$, and $\lambda_k=X_{k,1}$. Then, $(u,v)$-th entry of the transition matrix at time $\tau$ has absolute value one, i.e.,

\begin{equation*}
    \begin{split}
        \left| \mathcal{U}(\tau)_{u,v} \right|=\left|\sum_{k=1}^ne^{\iota \tau \lambda_k}P_{u,k}P_{v,k} \right| &= 1 \\
        \implies \left| \sum_{k=1}^n(-1)^{\braket{b(\sigma)|b(k)}}e^{\iota \tau \lambda_k} \right| &= n.
    \end{split}
\end{equation*}
This implies,
\begin{equation*}
    \begin{split}
        (-1)^{\braket{b(\sigma)|b(k)}}e^{\iota \tau \lambda_k} &= (-1)^{\braket{b(\sigma)|b(1)}}e^{\iota \tau \lambda_1},\; 1\leq k\leq n. \\
        \implies e^{\iota\tau(\lambda_k-\lambda_1)} &= (-1)^{\braket{b(\sigma)|b(k)}},
    \end{split}
\end{equation*}
which gives,
\begin{equation*}
    \begin{split}
        \tau(\lambda_k-\lambda_1) = \begin{cases} 2m_k\pi,& \mbox{if }\braket{b(\sigma)|b(k)} \mbox{ is even}. \\ (2m_k+1)\pi,& \mbox{if }\braket{b(\sigma)|b(k)}\mbox{ is odd}. \end{cases}
    \end{split}
\end{equation*}
Since, eigenvalues are integers and $\lambda_k-\lambda_1$ is even for each $k$,  $2\tau$ divides $\pi$ and we get
\begin{equation}\label{eq:sigma}
    \frac{\lambda_k-\lambda_1}{2} = \begin{cases} 2m_k,& \mbox{if }\braket{b(\sigma)|b(k)} \mbox{ is even}, \\ (2m_k+1),& \mbox{if }\braket{b(\sigma)|b(k)}\mbox{ is odd}, \end{cases}
\end{equation}
where $m_k$, $1\leq k\leq n$, is an integer. The value of $\sigma$ satisying Eq.~\ref{eq:sigma} can be computed by solving the equation for $k=2^j+1$, $0\leq j\leq d-1$, because $(d-j)$-th bit of $b(\sigma)$ is 1 if and only if $\frac{\lambda_{k}-\lambda_1}{2}$ is odd for $k=2^j+1$. This implies $\sigma$ exists uniquely. Now, if $\sigma=1$ then $u=v$, which implies the graph is periodic at $u$ with period dividing $\frac{\pi}{2}$, and if $\sigma\neq 1$ then $u\neq v$, which implies the graph admits PST in time $\frac{\pi}{2}$.
\end{proof}

In Theorem~\ref{thm:pst}, if $\sigma=1$ then the associated graph is periodic with period $\frac{\pi}{2}$, and if $\sigma\neq 1$ then a pair $\{u,v\}$ is a PST pair, with time $\frac{\pi}{2}$, in the associated graph if and only if $b(\sigma)=b(u)\oplus b(v)$. Notice that, we have nowhere used the fact that the first entry of Z is zero, therefore, the result holds true even if the weighted cubelike graphs have loops. 

\begin{remark}
We can calculate the value of $\sigma$ from $Z$ itself, by applying Theorem~\ref{thm:pst} over cubelike graphs. From Eq.~\ref{eq:eigen}, we get
\[
\begin{split}
    \lambda_k-\lambda_1 &= \sum_{l=1}^n\left[(-1)^{\braket{b(k)|b(l)}}-1 \right]Z_{l,1}, \\
    \implies \frac{\lambda_1-\lambda_k}{2} &= \sum_{l\in O_k}Z_{l,1},
\end{split}
\]
where $O_k=\{l:1\leq l \leq n,\; \braket{b(k)|b(l)}\mbox{ is odd}\}$. Using Eq.~\ref{eq:sigma}, we get 
\[\braket{b(\sigma)|b(k)}\mbox{ is odd } \mbox{ if and only if } \sum_{l\in O_k}Z_{l,1} \mbox{ is odd}.\]
We can further restrict the values of $k$ to $2^j+1,$ where $0\leq j \leq d-1$, and compute $b(\sigma)$.
\end{remark}

\subsection{Numerical results}
We illustrate perfect state transfer in weighted cubelike graphs via Table~\ref{tab:PST}, where weights are integers. We assign random integer entries in $Z$ and study perfect state transfer in the associated weighted cubelike graph with time $\frac{\pi}{2}$. In every graph tested so far, the minimum time at which perfect state transfer occurs is $\frac{\pi}{2}$, and if perfect state transfer does not occur then the graph is periodic with minimum period $\frac{\pi}{2}$. In each case, if the graph admits perfect state transfer then the vertex set is partitioned into PST pairs such that if $\{u,v\}$ and $\{x,y\}$ are two PST pairs, then $b(u)\oplus b(v)=b(x)\oplus b(y)$.

    
    \begin{longtable}{|p{.5cm}|p{3cm}|p{3cm}|p{2cm}|}
    \caption{An illustration of PST in weighted cubelike graphs.\label{tab:PST}}\\
        \hline
        S.I. & Z (Fixed edges) & X (Eigenvalues) & PST pairs \\
         \hline
        1 & {[}0, 1, {-}7, {-}10{]} & {[}{-}16, 2, 18, {-}4{]} & (1, 4), (2, 3) \\
         \hline
        2 & {[}0, 50, {-}10, {-}3{]} & {[}37, {-}57, 63, {-}43{]} & (1, 4), (2, 3) \\
        \hline
        3 & {[}0, 3, 1, 4, {-}6, 0, {-}1, 10{]} & {[}11, {-}23, {-}17, 5, 5, 11, 13, {-}5{]} & (1, 6), (2, 5), (3, 8), (4, 7) \\
        \hline
        4 & {[}0, 2, 3, 4, 5, 6, 5, 4{]} & {[}29, {-}3, {-}3, {-}3, {-}11, {-}3, {-}7, 1{]} & (1, 1), (2, 2), (3, 3), (4, 4), (5, 5), (6, 6), (7, 7), (8, 8) \\
        \hline
        5 & {[}0, {-}72, 38, 93, 100, {-}86, {-}91, {-}42{]} & {[}{-}60, 154, {-}56, 362, 178, {-}120, {-}350, {-}108{]} & (1, 6), (2, 5), (3, 8), (4, 7) \\
        \hline
        6 & {[}0, 5, {-}1, {-}4, {-}1, 5, 3, 2, {-}8, 10, {-}8, {-}4, {-}8, 1, 7, {-}1{]} & {[}{-}2, {-}30, 10, {-}46, {-}18, {-}18, 38, 2, 20, 16, 8, 16, 0, 24, {-}16, {-}4{]} & (1, 9), (2, 10), (3, 11), (4, 12), (5, 13), (6, 14), (7, 15), (8, 16) \\
        \hline
        7 & {[}0, {-}83, {-}80, {-}35, 65, 64, {-}31, {-}50, 94, 5, 97, {-}60, {-}92, {-}25, {-}5, 24{]} & {[}{-}112, 208, 168, 4, {-}12, 360, 20, 116, {-}188, {-}92, 316, 216, {-}480, {-}324, {-}376, 176{]} & (1, 1), (2, 2), (3, 3), (4, 4), (5, 5), (6, 6), (7, 7), (8, 8), (9, 9), (10, 10), (11, 11), (12, 12), (13, 13), (14, 14), (15, 15), (16, 16)\\
        \hline
        8 & {[}0, {-}30, 99, 5, 46, {-}85, {-}19, 100, 83, {-}10, {-}43, {-}4, 59, 60, 29, 22{]} & {[}312, 196, {-}66, 310, {-}112, 160, 38, {-}174, {-}80, 76, {-}442, 62, 176, 64, {-}66, {-}454{]} & (1, 3), (2, 4), (5, 7), (6, 8), (9, 11), (10, 12), (13, 15), (14, 16) \\
        \hline
        9 & {[}0, {-}10, {-}5, 0, {-}7, {-}7, {-}5, 2, {-}1, {-}3, {-}3, {-}9, {-}7, 3, 6, {-}8, {-}5, 5, 0, 4, 3, 9, 2, 10, 1, 7, 8, {-}3, 8, {-}3, {-}2, {-}8{]} & {[}{-}18, 4, 4, {-}22, {-}10, 4, 0, {-}2, 10, {-}64, {-}44, 58, {-}26, 20, 4, 2, {-}90, 16, {-}24, 10, {-}6, 28, 32, 58, {-}30, 36, 0, 42, 50, {-}4, {-}12, {-}26{]} & (1, 4), (2, 3), (5, 8), (6, 7), (9, 12), (10, 11), (13, 16), (14, 15), (17, 20), (18, 19), (21, 24), (22, 23), (25, 28), (26, 27), (29, 32), (30, 31) \\
        \hline
    \end{longtable}

\section{Discussion and Future work}
We classified a graph family with respect to their eigenvectors, i.e., their eigenvectors are identical, where eigenvectors are same as that of cubelike graphs. The weights are integers and the corresponding eigenvalues obtained are also integers. We look forward to find other graph families having same eigenvectors and see if one member admits perfect state transfer then others also allow the same. Similarly, we can test the existence of perfect state transfer for discrete-time coined quantum walk on graphs having same eigenvectors.

\bibliographystyle{acm}
\bibliography{PSTWeightedCubelikeGraph}

\end{document}